\documentclass[10pt]{amsart}
\usepackage{enumerate}
\usepackage{amssymb,amsthm,amsmath}
\usepackage[numbers,sort&compress]{natbib}
\usepackage{color}
\usepackage{graphicx}
\usepackage{tikz}

\title[ ]{Arithmetic Spectral Transitions  for the Maryland Model }
\author{ Svetlana Jitomirskaya}
\address[ Svetlana Jitomirskaya]{ Department of Mathematics, University of California, Irvine, California 92697-3875, USA}
\email{szhitomi@math.uci.edu}

\author{Wencai Liu}
\address[Wencai Liu]{School of Mathematical Sciences, Fudan University, Shanghai 200433, P. R. China} \email{liuwencai1226@gmail.com}
\address{Current address: Department of Mathematics, University of California, Irvine, California 92697-3875, USA}

\theoremstyle{plain}
\newtheorem{theorem}{Theorem}[section]
\newtheorem{corollary}[theorem]{Corollary}
\newtheorem{lemma}[theorem]{Lemma}
\newtheorem{proposition}[theorem]{Proposition}

\theoremstyle{definition}

\newtheorem{remark}[theorem]{Remark}

\begin{document}


\begin{abstract}
   We give a precise description of spectra of the
  Maryland model  $ (h_{\lambda,\alpha,\theta}u) _n=u_{n+1}+u_{n-1}+
  \lambda \tan \pi(\theta+n\alpha)u_n$ for all values of parameters.

We introduce an
  arithmetically defined index $\delta (\alpha, \theta)$ and show that
  for $\alpha\notin\mathbb{Q},\,$
 $\sigma_{sc}(h_{\lambda,\alpha,\theta})=\overline{\{e:\gamma_{\lambda}(e) <\delta (\alpha, \theta) \}}$ and
 $\sigma_{pp}(h_{\lambda,\alpha,\theta})=\{e:\gamma_{\lambda}(e) \geq \delta (\alpha, \theta) \}$.
Since  $\sigma_{ac}(h_{\lambda,\alpha,\theta})=\emptyset,\;$
  this gives  complete description of the spectral
  decomposition for {\it all} values of parameters
  $\lambda,\alpha,\theta$, making it the first case of a family where
  arithmetic spectral transition is described without any parameter exclusion.
The set of eigenvalues can be explicitly identified for all parameters,
using the {\it quantization condition}. We also establish, for the
first time for this or any other model, a quantization condition for singular continuous
spectrum (an arithmetically defined measure zero set that supports
singular continuous measures) for all parameters.


\end{abstract}

\maketitle
\parskip -1pt
\section{Introduction}
The Maryland model
is a  
discrete self-adjoint 
Schr\"{o}dinger operator on $\ell^2(\mathbb{Z})$ of the form
 \begin{equation}\label{MME}
(h_{\lambda,\alpha,\theta}u) _n=u_{n+1}+u_{n-1}+  \lambda \tan \pi(\theta+n\alpha)u_n.
 \end{equation}
We will refer to $\lambda$ as coupling, $\alpha$ as frequency, and
$\theta $ as the phase.  For simplicity, we will sometimes omit the
dependence on parameters $ \lambda,\alpha$ and $\theta$ in notations.
 \par
This operator was proposed by Grempel, Fishman, and Prange in 1982
\cite{grempel1982localization} as a model stemming from the study of
quantum chaos. They exactly computed its integrated
density of states and obtained, in an essentially
rigorous way, a dense set of explicitly determined eigenvalues,
corresponding to exponentially decaying eigenfunctions, for
Diophantine frequencies. Even though artificial,  it
became quite popular in Physics as an exactly solvable example of the family of
incommensurate models, e.g. \cite{berry, scholar}, with spectral
transitions governed by arithmetics. It was recently identified as
the topological quantum phase transition point of  a certain family \cite{new}. Mathematically, it is
interesting due to its richness of spectral theory, abundance of unusual
features, yet amenability to analysis that often has a tendency to become relevant to other potentials as well. It was dubbed the
Maryland model
by B. Simon \cite{MR776654,cyconschrodinger} who  cited the original
work \cite{grempel1982localization}  as a textbook
example of dealing with small divisors. Maryland model is the subject
of Sec. 10.3 in \cite{cyconschrodinger}  and  Ch. 18 in \cite{pastur1992spectra}.

 When $\alpha \in \mathbb{Q},$ operator $h$ is periodic, so has purely absolutely continuous
 spectrum for all $\theta.$ On the other hand, for $\alpha
 \notin \mathbb{Q}$  it has unbounded potential, thus no absolutely continuous
 spectrum for all $\theta$ \cite{simon1989trace}.
As typical for quasiperiodic-type
 potentials it undergoes a transition from Anderson localization (pure
 point spectrum with exponentially decaying eigenfunctions) to
 singular-continuous spectrum as $\alpha$ changes from Diophantine to
 Liouville. Both Diophantine (for all $\theta$ ) and Liouville (for a.e. $\theta$) parts of the above sentence
 were established in   \cite{MR776654,MR767188}, by
 different methods, thus solving the problem of determining the
 spectral type for a.e. $\alpha,\theta$. The intermediate regime of
 $\alpha$ was studied for a.e. $\theta$, and the location of a.e.
 $\theta$  transition in $\alpha$ was conjectured already in
 \cite{MR776654}, but the regime of the neighborhood of the transition
 remained open. Note that typically, in families with transitions,
 neighborhood of the transition both represents  the most challenge and
 is of the most interest. Also the location (or existence) of the
 transition in $\theta$ was not even conjectured previously. In this paper we determine the spectral types for all $\alpha,\theta,$ thus also precisely describing the associated
 arithmetical metal-insulator transitions. As far as we know, this is
 the first time spectral properties of a quasiperiodic operator are described for {\it all} parameters
 in a situation with arithmetic spectral transitions.

We will assume  $\alpha\in \mathbb{R}\backslash \mathbb{Q}$ (otherwise
the spectrum is purely absolutely continuous),
$\lambda > 0$ (otherwise use
$h_{\lambda,\alpha,\theta}=h_{-\lambda,-\alpha,-\theta}$) and $
\theta\notin  1/2 + \alpha\mathbb{Z}+\mathbb{Z}$ (so the potential is
well defined). In the rest of the paper we will often use ``all
$\lambda (\alpha,\theta)$'', meaning ``all $\lambda (\alpha,\theta)$ as above''.

Let $ \frac{p_n}{q_n} $ be  the continued fraction approximants   to
$\alpha.$
Define an index $\delta
(\alpha, \theta)\in [-\infty,\infty]$,
 \begin{equation}\label{Def.delta_1}
    \delta (\alpha, \theta)=\limsup_{n\rightarrow\infty} \frac{ \ln||q_n(\theta-1/2)||_{\mathbb{R}/\mathbb{Z}}+\ln q_{n+1}}{q_n},
 \end{equation}
 where $||x||_{\mathbb{R}/\mathbb{Z}}=\min_{\ell \in
   \mathbb{Z}}|x-\ell| $. Let  $\gamma_\lambda(e)$ be the Lyapunov
 exponent (see (\ref{LEBiaoda}) (\ref{Def.le}))
, and $ k_\lambda(e)$  the IDS of the Maryland model  (see
 (\ref{IDS})). Let
$\sigma_{pp}(h)$,  $\sigma_{sc}(h)$ and
     $\sigma_{ac}(h)$ be the
 pure point spectrum,
   sc spectrum, and ac spectrum  of $h,$ respectively.  Our main result is

\begin{theorem}\label{Maintheorem}
For all $\alpha,\lambda,\theta$, we have
\begin{enumerate}
  \item $\sigma_{sc}(h_{\lambda,\alpha,\theta})=\overline{\{e:\gamma_{\lambda}(e) <\delta (\alpha, \theta) \}}$
\item $\sigma_{pp}(h_{\lambda,\alpha,\theta})=\{e:\gamma_{\lambda}(e)
  \geq \delta (\alpha, \theta) \}$
\item $\sigma_{ac}(h_{\lambda,\alpha,\theta})=\emptyset.$



\end{enumerate}
\end{theorem}
\vskip .5 in

\begin{tikzpicture}
\draw [->](-3.5,0)--(3.5,0);
\draw [->](0,0)--(0,4.5);
\draw [dashed] (-1.5,0)--(-1.5,1.6);
\draw [dashed](1.5,0)--(1.5,1.6);
\draw (-2.7,1.6)--(2.7,1.6);
\draw plot [smooth] coordinates {(-2.8,2.4)(-1.5,1.6)(0,0.7)(1.5,1.6)(2.8,2.4)};
\node [below] at (2.9,0){$e$};
\node [right] at (0,4.3){$\gamma_{\lambda}(e)$};
\node [right] at (2.7,1.6){$\delta(\alpha,\theta)$};
\node [above] at (2.2,0){$ \sigma_{pp}$};
\node [above] at (-2.2,0){$ \sigma_{pp}$};
\node [above] at (0,0){$ \sigma_{sc}$};
\node [below] at (0,-0.3){Fig.1};

\end{tikzpicture}
\begin{tikzpicture}
\draw [->](-3.5,0)--(3.5,0);
\draw [->](0,0)--(0,4.5);
\draw (-2.7,0.7)--(2.7,0.7);
\draw plot [smooth] coordinates {(-2.8,2.4)(-1.5,1.6)(0,0.7)(1.5,1.6)(2.8,2.4)};
\node [below] at (2.9,0){$e$};
\node [right] at (0,4.3){$\gamma_{\lambda}(e)$};
\node [right] at (2.7,0.7){$\delta(\alpha,\theta)$};
\node [above] at (0,0){$ \sigma_{pp}$};
\node [below] at (0,-0.3){Fig.2};

\end{tikzpicture}
\begin{tikzpicture}
\draw [->](-3.5,0)--(3.5,0);
\draw [->](0,0)--(0,4.5);
\draw (-2.7,0.5)--(2.7,0.5);
\draw plot [smooth] coordinates {(-2.8,2.4)(-1.5,1.6)(0,0.7)(1.5,1.6)(2.8,2.4)};
\node [below] at (2.9,0){$e$};
\node [right] at (0,4.3){$\gamma_{\lambda}(e)$};
\node [right] at (2.7,0.5){$\delta(\alpha,\theta)$};
\node [above] at (0,0){$ \sigma_{pp}$};
\node [below] at (0,-0.3){Fig.3};

\end{tikzpicture}


\parskip 0pt
\begin{remark}
(3) is of course listed for completeness only, as it follows directly
from \cite{simon1989trace}.
 \end{remark}

   \begin{remark} This gives the following more precise description:\\
  \begin{description}
    \item[case 1]   If $\delta (\alpha, \theta)=+\infty$, then Theorem
      \ref{Maintheorem} implies $h$ has purely singular continuous
      spectrum with
      $\sigma_{sc}(h_{\lambda,\alpha,\theta})=(-\infty,\infty)$.
    \item[case 2]  If $\gamma_{\lambda}(0)<\delta (\alpha, \theta)<+\infty$ (see Fig.1),  it is clear that $\overline{\{e:\gamma_{\lambda}(e) <\delta (\alpha, \theta) \}}=\{e:\gamma_{\lambda}(e) \leq \delta (\alpha, \theta) \}$. Then

         \begin{equation*}
         \sigma_{sc}(h_{\lambda,\alpha,\theta})=\{e:\gamma_{\lambda}(e)\leq\delta (\alpha, \theta) \}
  \end{equation*}
    and
  \begin{equation*}
  \sigma_{pp}(h_{\lambda,\alpha,\theta})=\{e:\gamma_{\lambda}(e) \geq \delta (\alpha, \theta) \} .
  \end{equation*}

    \item[case 3] If $\delta (\alpha, \theta)\leq
      \gamma_{\lambda}(0)$ (see Fig.2 and Fig.3), it is clear that
      $\overline{\{e:\gamma_{\lambda}(e) <\delta (\alpha, \theta) \}}=
      \emptyset$.  Then $h$ has pure point spectrum with
      $\sigma_{pp}(h_{\lambda,\alpha,\theta})=(-\infty,\infty)$.
  \end{description}
   \end{remark}


   \begin{remark}  For Diophantine $\alpha$ and any $\theta$, $\delta (\alpha, \theta)=0,$
thus we are always in case 3. For Liouville $\alpha$ and
a.e. $\theta$, $\delta (\alpha, \theta)=\infty,$ so we are in case
1. However, for every $\alpha$, no matter how Liouville, there
  exists a dense  set of $\theta$ as in Case 3, so with pure
  point spectrum;  see Corollary \ref{cannotextendtoalltheta}.
\end{remark}

The more precise history of this question is the following. The results of \cite{grempel1982localization} were made rigorous, by different methods, in
\cite{MR776654,MR767188} thus proving the Diophantine and
 Liouville (a.e. $\theta$) version of
Theorem~\ref{Maintheorem}. Moreover, the analysis of  Simon
\cite{MR776654}  went
deeper in the arithmetics of $\alpha.$ Let us define an index
$\beta(\alpha)\in [0,\infty]$ by

\begin{equation}\label{Def.beta}
   \beta(\alpha)=\limsup_{n\rightarrow\infty}\frac{\ln q_{n+1}}{q_n}.
 \end{equation}

Note that for Diophantine $\alpha$, $\beta(\alpha)=0,$ and for
Liouville, $\beta(\alpha)=\infty.$ Simon showed that for
a.e. $\theta$, $h$ has purely dense point spectrum on
$\{e:\gamma_{\lambda}(e)>\beta(\alpha)\}$ and purely singular
continuous spectrum on $\{e:\tilde{\gamma}_{\lambda}(e)
<\frac{1}{2}\beta(\alpha)\}$, for a certain
$\tilde{\gamma}_{\lambda}(e)>\gamma_{\lambda}(e).$
The latter was improved to $\{e:\gamma_{\lambda}(e)
<\frac{1}{2}\beta(\alpha)\}$ in \cite{pastur1992spectra}.  Therefore
those results left open the neighborhood of the transition: $\{e:\beta(\alpha)\ge\gamma_{\lambda}(e)
\ge\frac{1}{2}\beta(\alpha)\}$
and are only at the measure-theoretic level in
$\theta$'s.\footnote{ It should be
noted that while one can extract an explicit condition on the removed
sets of $\theta$ in \cite{MR776654}, it is intrinsically not possible
within the argument of Pastur-Figotin
\cite{MR767188,pastur1992spectra}.}

There have been a number of interesting studies of a multidimensional
generalization of Maryland model (already in \cite{MR776654,MR767188,bls}) and other related models, for example, surface Maryland
model (e.g. \cite{jakvsic1998spectrum}). Certain other interesting aspects of the classical
Maryland model were also explored, e.g. \cite{fedotov2013exact}. However, the state
of the art of the question of spectral decomposition remained as
described above.

In general, having Diophantine and Liouvillean approaches meet at the
transition (or non-transition) is known to be a delicate task. It has
been recently achieved in several notable quasiperiodic problems
\cite{avila2009ten,AvilaYouZhou}. Here we see our main accomplishment however in
obtaining the result for {\it all} $\theta$ and precisely describing the
corresponding transition.

Our next result is concerned with the {\it quantization condition}: the
explicit description of the support of the spectral measure of $h.$
The Maryland model is often called exactly solvable. Indeed,
Fishman, Grempel, and Prange computed explicitly the Lyapunov exponent
$\gamma_\lambda(e)>0$ as satisfying
\begin{equation}\label{LEBiaoda}
     4\cosh \gamma_\lambda(e)=  \sqrt{ (2+e)^2 +\lambda^2} +\sqrt{ (2-e)^2 +\lambda^2},
 \end{equation}
and the integrated
density of states $k_{\lambda}(e)$ (IDS) (see \cite{MR700145} for the
definition of IDS) as
\begin{equation}\label{IDS}
    k_{\lambda}(e)=\frac{1}{2} +\frac{1}{\pi }\arctan (\ e /| \lambda |\tanh \gamma_{\lambda}(e) ).
\end{equation}
(for all $\lambda,\alpha$)\footnote{It may be interesting to think of
  the corresponding cocycle as monotonic. Even
  though, due to the discontinuity
  the theorems of \cite{mon} do not apply, some of the conclusions
  still hold.} and presented an explicit dense countable set
of $e$ that are eigenvalues of $ h_{\lambda,\alpha,\theta}$ with
Diophantine $\alpha :$ energies $e$ such that
\begin{equation} \label{q} k_\lambda(e)\in \theta-1/2+   \alpha\mathbb{Z}
   +\mathbb{Z}.
\end{equation}

 This was made rigorous in \cite{MR776654, MR767188}  (in particular, the completeness of
corresponding eigenfunctions was shown) and Simon dubbed (\ref{q}) the
{\it quantization condition} and extended it to the regime
$\{e:\gamma_{\lambda}(e)>\beta(\alpha)\}$. Thus in this regime the
supports of all spectral measures can be explicitly described through
(\ref{IDS}),(\ref{q}). Simon then posed a question of finding a
quantization condition (thus precise arithmetic description) for
supports of singular continuous spectral measures. Here we solve a
certain version of it. To formulate this
more precisely, let $\mu_{\lambda,\alpha,\theta}=\frac{1}{2}(\mu_{\lambda,\alpha,\theta}^{\delta_{-1}}+\mu_{\lambda,\alpha,\theta}^{\delta_0})$, where
$\mu_{\lambda,\alpha,\theta}^{\delta_i}$ is the spectral measure of operator
  $ h_{\lambda,\alpha,\theta}$ with respect to vector $\delta_i$, $i=-1,0$
($\delta_i$
is the Dirac mass at $i\in\mathbb{Z}$). A set
$A$ is a support of a singular measure
$\mu$ if $|A|=
\mu(A^c)=0,$ where
$|\cdot|$ stands for Lebesgue measure.  While
for a point measure it is natural to identify its support as the
collection of point masses, for singular continuous measure it is only
well defined up to sets of Lebesgue and $\mu$ measure zero.

Now, let  $$A_{\lambda,\alpha,\theta}= \{e:
||q_n(k_\lambda(e)-\theta-1/2)||\to 0 \;\mbox{as}\; n\to\infty\} .$$

Also let
   $$Q_{\lambda,\alpha,\theta}=\{e: k_\lambda(e)\in \theta-1/2+   \alpha\mathbb{Z}
   +\mathbb{Z}\}.$$ As $k_\lambda(e)$ is continuous  and strictly monotone, $Q_{\lambda,\alpha,\theta}$
   is a countable dense subset of $\mathbb{R},$ and
   $A_{\lambda,\alpha,\theta}$ is dense in $\mathbb{R}.$ Also,
   $|\{a\in[0,1]:||q_n(a-\theta-1/2)||<\epsilon\}|<2\epsilon,$ so we have that
   $|\{a\in[0,1]:||q_n(a-\theta-1/2)||\to 0\}|=0$ and thus
   $|A_{\lambda,\alpha,\theta}|=0 $ by the fact
   $\frac{d k_\lambda(e)}{de}>0.$

Let $\sigma_p(h)$ be the collection of all eigenvalues of $h.$ We have

\begin{theorem} \label{quant} For all $\lambda,\alpha,\theta$ as above
 \begin{enumerate}
\item $\mu_{\lambda,\alpha,\theta}$ is supported on $A_{\lambda,\alpha,\theta}$,
\item  the set of eigenvalues with corresponding eigenfunctions
  decaying exponentially coincides with  the set
  $\{e:\gamma_{\lambda}(e) >\delta (\alpha, \theta) \}\cap Q_{\lambda,\alpha,\theta}$,
\item $\{e:\gamma_{\lambda}(e)
>\delta (\alpha, \theta) \}\cap
Q_{\lambda,\alpha,\theta}\subseteq\sigma_{p}(h_{\lambda,\alpha,\theta})\subseteq
\{e:\gamma_{\lambda}(e)
\geq\delta (\alpha, \theta) \}.$
\end{enumerate}
\end{theorem}
\begin{remark}
Thus
$\sigma_{p}(h_{\lambda,\alpha,\theta})$ may differ from $\{\{e:\gamma_{\lambda}(e)
>\delta (\alpha, \theta) \}\cap Q_{\lambda,\alpha,\theta}\}$ only by a
subset of
$\{e:\gamma_{\lambda}(e)=\delta (\alpha, \theta) \}.$ Note  that
$\{e:\gamma_{\lambda}(e)=\delta (\alpha, \theta) \}$  consists of at most
two points, see Fig. 1-3.
\end{remark}

Theorem \ref{quant} therefore extends the quantization condition for
point spectrum to the entire region $\{e:\gamma_{\lambda}(e) >\delta
(\alpha, \theta) \} $ which is the entire set where point spectrum can
possibly be located up
to at most two points, and  establishes a version of the quantization condition for supports of
singular continuous measures: an arithmetically defined measure zero
set that supports $\mu_{\lambda,\alpha,\theta}$. To the best of our
knowledge, the latter is the first such arithmetic condition for
supports of singular conrinuous spectra, for any model\footnote{While an arithmetic
condition for supports of singular continuous spectra appearing for
Diophantine $\alpha$ and exceptional phases for the almost Mathieu
operators can be extracted from the proof in \cite{jems}, it does not
currently cover all parameters}. We believe that the methods of this
work could lead to obtaining information on supports of singular
continuous spectral measures for many other quasiperiodic operators
that in turn may allow dimensional analysis that is not intrinsically one-dimensional.

Of course, since supports are defined up to $\mu$- measure zero sets,
this could ostensibly be improved by finding smaller sets that have
the same properties. Indeed, we in fact prove a stronger version of
part 1, finding arithmetically defined $B_{\lambda,\alpha,\theta}$
that are proper subsets of $A_{\lambda,\alpha,\theta}$ and support
$\mu_{\lambda,\alpha,\theta}$. We formulate part 1 the way we do
(rather than a stronger statement that we prove) purely for aesthetic
purposes.

A well known support of $\mu_{\lambda,\alpha,\theta}$ is
$C_{\lambda,\alpha,\theta}:=\{e :$ such that there exists a
polynomially bounded solution to $h_{\lambda,\alpha,\theta}u=eu\}$.  To
find an arithmetic description of $C_{\lambda,\alpha,\theta}$ remains
an interesting open problem.

Finally, because of this aspect even after Theorem \ref{quant} and
absence of eigenvalues in $\{e:\gamma_{\lambda}(e)
<\delta (\alpha, \theta) \}$ is
established,  case (1) of Theorem \ref{Maintheorem} does not follow because it is
not clear that supports of singular continuous spectral measures {\it
  have to be}
dense in  $\{e:\gamma_{\lambda}(e) <\delta (\alpha, \theta)
\}$. Indeed, the fact that $\text{spec}(h)=\{e: k_{\lambda}(e+\varepsilon)-k_{\lambda}(e-\varepsilon) >0 \text{ for any } \varepsilon>0\}$
and (\ref{IDS}) only
 imply that the spectrum of Maryland model $
\sigma(h_{\lambda,\alpha,\theta})$ is $(-\infty,\infty)$   for almost every phase $\theta$ \cite{cyconschrodinger}.
For bounded potentials with some continuity it is easy to extend such result
to {\it all} $\theta$. Moreover,
 the fact that $
\sigma(h_{\lambda,\alpha,\theta}) \supseteq \{e:\gamma_{\lambda}(e)\geq\delta(\alpha,\theta)\}$  for all $\theta$
follows from the denseness of $Q_{\lambda,\alpha,\theta}.$ However, in the singular
continuous regime, denseness of $A_{\lambda,\alpha,\theta}$ does not
automatically yield the result. Still
\begin{theorem}\label{Spectrumforalltheta}
We have $\sigma(h_{\lambda,\alpha,\theta})=(-\infty,\infty)$ for all phases
$\theta.$
\end{theorem}

\par


The strategy in our proof of the main result is overall following that
 in Simon \cite {MR776654}, but we present novel ideas in both
localization and singular continuous parts.

From the localization side, our main contribution consists in a certain new
technique in handling the cohomological equation in the regime of very
small denominators: $\{e:\gamma_{\lambda}(e)<\beta(\alpha)\}$.

From the singular continuous side, we develop a version of Gordon-type
argument allowing to handle singular potentials in a sharp way. While
in the interests of brevity we present here an argument taylored to
the Maryland model, it is actually quite robust, and can be easily
modified to obtain similar results for a large class of
quasiperiodic-type models. Moreover, we believe our method has an even wider
applicability.

The rest of the paper is organized as follows:
 \par
 In \S2, we present a Gordon-type method for cocycles with
 singularities and use it   to  prove that
 $h_{\lambda,\alpha,\theta}$ has purely singular continuous spectrum
 on  $\{e:\gamma_{\lambda}(e) <\delta (\alpha, \theta) \}$.
 In \S3, we show that $h_{\lambda,\alpha,\theta}$  has  only pure point spectrum
 on  $\{e:\gamma_{\lambda}(e) >\delta (\alpha, \theta) \}$.
 In \S4, we prove that $
 \sigma(h_{\lambda,\alpha,\theta})=(-\infty,\infty)$ for all
 parameters, proving Theorem \ref{Spectrumforalltheta} and
  Theorem \ref{Maintheorem}. We also show that for any
  $\alpha\in\mathbb{R}\backslash\mathbb{Q},$ there exist phases with
  pure point spectrum.
In \S 5,  we turn to another aspect and 
 compute the  complexified Lyapunov exponent using Avila's technique \cite{avila2009global}. Besides being a new (and
 very simple) computation of $\gamma_\lambda(e)$, it gives another
 example of how regularity can coexist with positive Lyapunov exponent
 on the spectrum of an operator with singular potential.

\section{Singular continuous spectrum }
  In this section, we prove the following theorem.
 \begin{theorem}\label{maintheorem1}
Let
   $\delta (\alpha, \theta)$ be as in (\ref{Def.delta_1}),
 then $h_{\lambda,\alpha,\theta}$ has purely singular continuous spectrum
 on  $\{e:\gamma_{\lambda}(e) <\delta (\alpha, \theta) \}$. In particular,
 for almost every phase $\theta$ (only depends on $\alpha$, not on $\lambda$ or $e$), $h$
 has purely singular continuous spectrum
 on  $\{e:\gamma_{\lambda}(e) <\beta(\alpha)\}$.
\end{theorem}
We may assume $\delta (\alpha, \theta) >0$, otherwise  the set $\{e:\gamma_{\lambda}(e) <\delta (\alpha, \theta) \}$ is empty.
 If $\delta (\alpha,\theta)=\infty$, select $\delta$ to be any finite number bigger than $\gamma_{\lambda}(e)$.
Before giving the proof of Theorem \ref{maintheorem1}, some  preliminaries are necessary.

The one-step transfer-matrix of $hu=eu$ is given by

 \begin{equation*}
   A(\theta)=\left(
          \begin{array}{cc}
            e-\lambda\tan\pi \theta &  -1 \\
            1 & 0 \\
          \end{array}
        \right).
 \end{equation*}
So that $A(\theta+n \alpha) \left( \begin{array}{c}
                               u_n \\
                               u_{n-1}
                             \end{array}
                             \right)
= \left(\begin{array}{c}
                               u_{n+1} \\
                               u_{n}
                             \end{array} \right)$.
We will refer to the pair $(\alpha, A)$ as Maryland  cocycle
understood as a linear skew-product
$(\theta,w) \mapsto (\theta+\alpha,A(\theta) \cdot w)$, $\theta\in  \mathbb{R}/\mathbb{Z}, \theta\notin  1/2 + \alpha\mathbb{Z}+\mathbb{Z}$, $w \in
 \mathbb{R}^2$.  In general, for an invertible cocycle $(\alpha, B)$ we define
$B_n$ by $(\alpha,B)^n=(n\alpha,
B_n)\,n\in\mathbb{Z}$ so that for $n\geq 0,$

\begin{equation*}
     B_n(\theta) = B(\theta+(n-1)\alpha)B(\theta+(n-2)\alpha)\cdots B(\theta),
\end{equation*}
and   $B_{-n}(\theta)=B_{n}(\theta-n\alpha)^{-1}$.
\par
Given a cocycle $(\alpha,B)$,  the  Lyapunov exponent (LE) $L(\alpha,B)$ is defined by the
 following formula,
 \begin{equation}\label{Def.le}
   L(\alpha,B)=\lim_{n\rightarrow\infty} \frac{1}{n}\int_{\mathbb{T} } \ln \|B_n(\theta)\|d\theta,
 \end{equation}
Denote $\gamma _{\lambda}(e)= L(\alpha,A)$ (see Theorem \ref{Thmvaluele}).

 Because $\tan\pi \theta$ is singular at $\theta=\frac{1}{2}\mod
\mathbb{Z}$, we will rewrite the Maryland cocycle
as

\begin{equation*}
    \left(
                         \begin{array}{cc}
                           e- \lambda \tan\pi \theta&-1 \\
                           1 &0 \\
                         \end{array}
                       \right)=
                       \frac{1}{\cos \pi \theta}
                       \left(
                         \begin{array}{cc}
                           e\cos\pi \theta - \lambda \sin\pi \theta&-\cos\pi \theta \\
                           \cos\pi \theta&0 \\
                         \end{array}
                       \right).
 \end{equation*}
 Let  \begin{equation}\label{MMR}
 D(\theta) = \left(
                         \begin{array}{cc}
                           e\cos\pi \theta - \lambda \sin\pi \theta&-\cos\pi \theta \\
                           \cos\pi \theta&0 \\
                         \end{array}
                       \right),
                        \end{equation}
                       we have
                       \begin{equation}\label{MMCLE}
                        L(\alpha,A)=
                        L(\alpha,D)-\int_{\mathbb{T}}\ln|\cos\pi
                        \theta |d\theta.
                       \end{equation}
Notice that
\begin{equation}\label{ln2}
    \int_{\mathbb{T}}\ln|\cos\pi
                        \theta |d\theta=-\ln 2.
\end{equation}

 Thus $D(\theta)$ is the regular part and $\frac{1}{\cos\pi \theta}$ is
the singular part of the cocycle.
We now give the estimate of the two parts separately.
First, since $D(\theta)$ is analytic, we have that $\ln ||D_n(\theta)||$ is
a continuous subadditive cocycle, thus it is well known that, as a
corollary of unique ergodicity,
\begin{equation}\label{EstimateD}
  \lim_{n\rightarrow \infty}\frac{1}{n}\ln ||D_n(\theta)||\leq L(\alpha,D),
\end{equation}
uniformly in $\theta$.

In order to estimate the singular part we will need

\begin{lemma} (\text{Lemma } 9.7, \cite{avila2009ten})\label{Estimatecos}
Let $\alpha\in \mathbb{R}\backslash \mathbb{Q}$, $\theta\in\mathbb{R}$ and $0\leq \ell_0 \leq q_n-1$ be such that
$ |\cos\pi(\theta+\ell_0\alpha)|=\inf_{0\leq\ell\leq q_n-1}    |\cos\pi(\theta+\ell \alpha)|$, then for some absolute constant $C > 0$,
\begin{equation*}
    -C\ln q_n\leq \sum _{\ell=0,\ell\neq \ell_0}^{q_n-1} \ln|\cos\pi(\theta+\ell\alpha )|+(q_n-1)\ln2\leq  C\ln q_n.
\end{equation*}
\end{lemma}
We will also use that the denominators of continued fraction
approximants of $\alpha$ satisfy
\begin{equation}\label{Approximate1}
\| k\alpha\|_{\mathbb{R}/\mathbb{Z}}\geq
||q_n\alpha||_{\mathbb{R}/\mathbb{Z}},  1\leq k <q_{n+1},
\end{equation}
and
\begin{equation}\label{Approximate2}
      \frac{1}{2q_{n+1}}\leq\Delta_n\triangleq \|q_n\alpha\|_{\mathbb{R}/\mathbb{Z}}\leq\frac{1}{q_{n+1}}.
\end{equation}
We are now ready to prove
\begin{theorem}\label{Thmestimatec_j}
 For any $\varepsilon>0$, there exists a subsequence $\tilde{q}_k$ of $q_n$ such that
  the following estimate holds
\begin{equation}\label{fur2}
 \prod_{j=0}^{\tilde{q}_k-1} |c_j|\geq \frac{e^{( \delta(\alpha,\theta)   -\ln2-\varepsilon)\tilde{q}_k}}{\tilde{q}_{k+1} },
\end{equation}
where $c_j = \cos\pi(\theta+j\alpha) $.

\end{theorem}

 \begin{proof}

 By the definition of $\delta (\alpha,\theta)$ (\ref{Def.delta_1}), for any small enough  $\varepsilon>0$,  there  exists a subsequence $\tilde{q}_k$ of $q_n$ such that
  \begin{equation}\label{q_ktheta}
    ||\tilde{q}_k(\theta-1/2)||_{\mathbb{R}/\mathbb{Z}}\geq \frac{e^{(\delta - \varepsilon/8)\tilde{q}_k}}{\tilde{q}_{k+1} }.
  \end{equation}

   This implies
  \begin{equation}\label{q_k+1}
   \tilde{q}_{k+1}\geq e^{(\delta - \varepsilon/8)\tilde{q}_k}.
  \end{equation}
   Let $|c_{j_0}|$  be the smallest one of $|c_j|$,$\,j=0,1,\cdots, \tilde{q}_k-1$.
  By Lemma \ref{Estimatecos} one has
  \begin{equation*}
 \prod_{j=0}^{\tilde{q}_k-1} |c_j|\geq |c_{j_0}|e^{ -\tilde{q}_k(\ln2+\varepsilon/4)}.
\end{equation*}
 In order to prove (\ref{fur2}),  it suffices to prove that
 \begin{equation}\label{Def.c_0}
   |c_{j_0}|\geq \frac{e^{(\delta -\frac{3}{4}\varepsilon)\tilde{q}_k}}{\tilde{q}_{k+1} },
 \end{equation}
 for $k$ large enough.
 \par
 If the estimate (\ref{Def.c_0}) does not hold, i.e.,
 $|c_{j_0}|\leq \frac{e^{(\delta -\frac{3}{4}\varepsilon)\tilde{q}_k}}{\tilde{q}_{k+1} }$, for some $0\leq j_0\leq \tilde{q}_k-1$, then
 \begin{equation*}
     ||  \theta-1/2 +j_0\alpha||_{\mathbb{R}/\mathbb{Z}}<\frac{e^{(\delta -\frac{1}{2} \varepsilon)\tilde{q}_k}}{\tilde{q}_{k+1} }.
 \end{equation*}
 This implies
 \begin{eqnarray*}
     || \tilde{q}_k( \theta-1/2)||_{\mathbb{R}/\mathbb{Z}} &\leq&  ||  j_0 \tilde{q}_k\alpha||_{\mathbb{R}/\mathbb{Z}}+  \frac{e^{(\delta -\frac{1}{2} \varepsilon)\tilde{q}_k}}{\tilde{q}_{k+1} }\\
     &\leq& \frac{\tilde{q}_k}{\tilde{q}_{k+1}} +\frac{e^{(\delta-\frac{1}{2} \varepsilon)\tilde{q}_k}}{\tilde{q}_{k+1} } \\
     &\leq& \frac{e^{(\delta -\frac{1}{4}\varepsilon)\tilde{q}_k}}{\tilde{q}_{k+1} }
 \end{eqnarray*}
 The second inequality holds by (\ref{Approximate2}).
This   contradicts (\ref{q_ktheta}).

\end{proof}

 After this preparatory work we can move to showing absence of point
spectrum by a Gordon-type method.
Next we always assume $\varepsilon>0$ is sufficiently small. The
following simple Lemma has become the classics of the
field.\footnote{see \cite{gor2} for its history.}
\begin{lemma}(\cite{gor1,MR776654,cyconschrodinger})\label{Lemgordonidea1}
Let $B\in \text{SL}(2,\mathbb{R})$ and  $x$ be a unit vector in $\mathbb{C}^2$,  then $$\max{\{||B^2x||, ||Bx||,||B^{-1}x||\}}\geq \frac{1}{4}.$$
\end{lemma}

It is usually used in conjunction with a perturbation estimate. We
will use
\begin{lemma}(\cite{MR776654} )\label{Lemgordonidea2}
Let $A^1,A^2,\cdots,A^n$ and  $B^1,B^2,\cdots,B^n$ be $2\times2$ matrices with $||\prod_{m=0}^{\ell-1}A^{j+m}||\leq C e^{d\ell}$ for some constant $C$ and $d$.
Then
\begin{equation*}
    ||(A^n+B^n)\cdots(A^1+B^1)-A^n\cdots A^1||\leq Ce^{dn} (\prod_{j=1}^n(1+Ce^{-d}||B^j||)-1).
\end{equation*}

\end{lemma}
\par
Classical method uses rational approximation to rule out the $\ell^2$ solutions. Here we use the Diophantine properties  of $\alpha$ directly
to show the transfer matrix  behaves like a periodic one. Our final
preparation is
\par
\begin{lemma}\label{Transfermatrixperiodic}
Suppose for any $j$ either  $ \eta^j=0 $ or $||\eta^j||\leq\frac{C}{c_j^2\tilde{q}_{k+1}}$ and $|c_j|\geq \frac{1}{C\tilde{q}_k}$. Then
 for $\varepsilon>0$, the following estimates hold if $k$ is sufficiently large:
\begin{equation}\label{Es1a}
||A_{\tilde{q}_k}(\theta)-\prod_{j=\tilde{q}_k-1}^{0} (A(\theta+j\alpha)+\eta^j)||\leq Ce^{-(\delta(\alpha,\theta)- \gamma_\lambda(e)-\varepsilon)\tilde{q}_k   },
\end{equation}
\begin{equation}\label{Es2a}
||A_{\tilde{q}_k}^{-1}(\theta)-\prod_{j=0}^{\tilde{q}_k-1}(A(\theta+j\alpha)^{-1}+\eta^j)||\leq Ce^{-(\delta(\alpha,\theta)- \gamma_\lambda(e)-\varepsilon)\tilde{q}_k   },
\end{equation}

\end{lemma}

\begin{proof} We only prove (\ref{Es1a}) with the case $|c_j|\geq \frac{1}{C\tilde{q}_k}$ and
$||\eta^j||\leq\frac{C}{c_j^2\tilde{q}_{k+1}}$ for all  $j$.   For the other cases, the proof is  similar. Below,    $C$ is a large constant and may  depend   on $e,\lambda,\alpha$.
\par

 We have
 \begin{eqnarray}
 \nonumber
 A_{\tilde{q}_k}(\theta)-\prod_{j=\tilde{q}_k-1}^{0}
 (A(\theta+j\alpha)+\eta^j)
   &=&  \prod_{j=\tilde{q}_k-1}^{0}(\frac{D^j  }{c_j}+\eta^j) -\prod_{j=\tilde{q}_k-1}^{0}\frac{D^j }{c_j}\\
     &=&  ( \prod_{j=0}^{\tilde{q}_k-1}\frac{1}{c_j})(  \prod_{j=\tilde{q}_k-1}^{0}(D^j+O( \frac{1 }{c_{j }\tilde{q}_{k+1}})) -\prod_{j=\tilde{q}_k-1}^{0}D^j ),\label{Estimateeta_j}
 \end{eqnarray}
 where $D^j=D(\theta+j\alpha)$.
 \par
 Denote  $L:=L(\alpha,D)$.   Applying Lemma \ref{Lemgordonidea2} to (\ref{Estimateeta_j}), one has
\begin{eqnarray}
  \nonumber
  || A_{\tilde{q}_k}(\theta-\tilde{q}_k\alpha) -A_{\tilde{q}_k}(\theta)|| &\leq& ( \prod_{j=0}^{\tilde{q}_k-1}\frac{1}{|c_j|})e^{(L+\varepsilon)\tilde{q}_k}
     (\prod_{j=0}^{\tilde{q}_k-1}(1+ \frac{C }{ |c_{j }|\tilde{q}_{k+1}}) -1)\\
      \nonumber
     &\leq& C  \frac{\tilde{q}_{k+1} }{e^{(\delta -\ln2-\varepsilon)\tilde{q}_k}}  e^{(L+ \varepsilon)\tilde{q}_k}
     ((1+ C\frac{ \tilde{q}_k }{  \tilde{q}_{k+1}})^{\tilde{q}_k} -1)
     \\ \nonumber
      &\leq& Ce^{-(\delta(\alpha,\theta)- \gamma_\lambda(e)-4 \varepsilon)\tilde{q}_k   }¡£
 \end{eqnarray}
 The first inequality holds by (\ref{EstimateD}), the second holds by
 (\ref{fur2}) and our
 assumption on the minimum of $|c_j|$, the third holds by  (\ref{MMCLE}), (\ref{ln2}) and (\ref{q_k+1}).
 \end{proof}

\textbf{Proof of Theorem \ref{maintheorem1}.}
We start with the proof  of the first part, i.e. that  $h_{\lambda,\alpha,\theta}$ has purely singular continuous spectrum
 on  $\{e:\gamma_{\lambda}(e) <\delta (\alpha, \theta) \}$. We know
 $\sigma_{ac}(h)=\emptyset$ so
assume $h_\theta$ has an  $\ell^2$ solution $u_n$. Without loss of generality, assume the vector  $ \left(\begin{array}{c}
                                                                                            u_0 \\
                                                                                            u_{-1}
                                                                                          \end{array}\right)
$
is unit.
\par
Let $\varphi(n)=\left(\begin{array}{c}
        u_n\\
       u_{n-1}
     \end{array}\right)
 $,
 then we have
 \begin{equation}\label{varphi_n}
    \varphi(n)=A_{n}(\theta) \varphi(0),
 \end{equation}
 and
 \begin{equation}\label{varphi_-n}
    \varphi(-n)=A_{n}^{-1} (\theta-n\alpha)\varphi(0).
 \end{equation}

For simplicity, denote   $\varphi :=\varphi(0)$.
 Applying Lemma \ref{Lemgordonidea1}, for any $k$, we have that  either $||A_{\tilde{q}_k}(\theta)\varphi||\geq \frac{1}{4}$, or $||A_{\tilde{q}_k}^{-1}(\theta)\varphi||\geq \frac{1}{4}$,
 or $||A_{\tilde{q}_k}^2(\theta)\varphi||\geq \frac{1}{4}$.
 \par
 Case  1: $||A_{\tilde{q}_k}(\theta)\varphi||\geq \frac{1}{4}$. Then  $||\varphi(\tilde{q}_k) ||\geq\frac{1}{4}$ by (\ref{varphi_n}).
 \par
 Case 2:
 $||A_{\tilde{q}_k}^{-1}(\theta)\varphi||\geq \frac{1}{4}$.  Let   $|c_{j_0}|$  be the smallest one of $|c_j|$, $j=0,1,\cdots, \tilde{q}_k-1$.
 Then  one has
 \begin{equation}\label{jneq j_o}
    |c_{j }|\geq \frac{1}{C\tilde{q}_k} \text{ for } j\neq j_0.
 \end{equation}

We now consider the difference between $A_{\tilde{q}_k}^{-1}(\theta)\varphi$ and $A_{\tilde{q}_k}^{-1}(\theta-\tilde{q}_k\alpha)\varphi $.
\par
Let
\begin{equation*}
     A_{\tilde{q}_k}^{-1}(\theta) -A_{\tilde{q}_k}^{-1}(\theta-\tilde{q}_k\alpha) =M_1 +M_2,
\end{equation*}
where

 \begin{eqnarray*}
    M_1&=&  A_{\tilde{q}_k}^{-1}(\theta)-A^{-1}(\theta-\tilde{q}_k\alpha)A^{-1}(\theta+\alpha-\tilde{q}_k\alpha)  \cdots   A^{-1}(\theta+(j_0-1)\alpha-\tilde{q}_k\alpha) \\
     & &   A^{-1}(\theta+j_0\alpha )   A^{-1}(\theta+(j_0+1)\alpha-\tilde{q}_k\alpha)\cdots
    A^{-1}(\theta+(\tilde{q}_k-1)\alpha-\tilde{q}_k\alpha)
 \end{eqnarray*}
 and
 \begin{eqnarray*}
   M_2 &=& A^{-1}(\theta-\tilde{q}_k\alpha)A^{-1}(\theta+\alpha-\tilde{q}_k\alpha)
      \cdots   A^{-1}(\theta+(j_0-1)\alpha-\tilde{q}_k\alpha)\\
    & &  \varepsilon_{j_0}  A^{-1}(\theta+(j_0+1)\alpha-\tilde{q}_k\alpha)\cdots
    A^{-1}(\theta+(\tilde{q}_k-1)\alpha-\tilde{q}_k\alpha),
 \end{eqnarray*}
 where $ \varepsilon_{j_0} =A^{-1}(\theta+j_0\alpha)-A^{-1}(\theta+j_0\alpha-\tilde{q}_k\alpha)$.
 \par
 Let
\begin{equation*}
\eta_j=\frac{D(\theta+j\alpha-\tilde{q}_k\alpha)}{\cos\pi(\theta+j\alpha-\tilde{q}_k\alpha)}-\frac{D(\theta+j\alpha)}{\cos\pi(\theta+j\alpha )},
\end{equation*}
 where   $D$ is given by (\ref{MMR}).
 It is easy to see that
\begin{equation*}
||D(\theta+j\alpha-\tilde{q}_k\alpha)-D(\theta+j\alpha)||\leq  \frac{C }{ \tilde{q}_{k+1}}
\end{equation*}
and
\begin{equation*}
   |\cos\pi(\theta+j\alpha-\tilde{q}_k\alpha)-\cos\pi(\theta+j\alpha )|\leq \frac{C }{ \tilde{q}_{k+1}},
\end{equation*}
 for all $|j|\leq \tilde{q}_k$, since $||\tilde{q}_k\alpha||_{\mathbb{R}/\mathbb{Z}}\leq \frac{1 }{ \tilde{q}_{k+1}}$ by (\ref{Approximate2}).
 Combining with (\ref{Def.c_0}), one has
 \begin{equation}\label{eta_j}
     ||\eta_j||=   
     O( \frac{1 }{c_j^2  \tilde{q}_{k+1}}).
 \end{equation}

Thus one has
 \begin{equation*}
     \varepsilon_{j_0} =\left(
                          \begin{array}{cc}
                           0& 0 \\
                            0 & O(\frac{1}{ c_{j_0}^2 \tilde{q}_{k+1} } ) \\
                          \end{array}
                        \right)
      .
 \end{equation*}
 By Lemma \ref{Transfermatrixperiodic} with $\eta_{j_0}=0$ (using (\ref{Es2a}) and (\ref{eta_j})),  we have
 \begin{equation*}
|| M_1\varphi ||\leq  Ce^{-(\delta(\alpha,\theta)- \gamma_\lambda(e)-  \varepsilon)\tilde{q}_k   }.
 \end{equation*}
  We will show that
   \begin{equation}\label{M_2}
|| M_2\varphi ||\leq   \frac{1}{16}.
 \end{equation}
This would imply that for $k$ large enough
 \begin{equation*}
    || A_{\tilde{q}_k}^{-1}(\theta)\varphi -A_{\tilde{q}_k}^{-1}(\theta-\tilde{q}_k\alpha)\varphi||\leq  \frac{1}{8}.
\end{equation*}
and therefore (using (\ref{varphi_-n}))
\begin{equation*}
   ||\varphi(-\tilde{q}_k )||=  ||  A_{\tilde{q}_k}^{-1}(\theta-\tilde{q}_k\alpha)\varphi||\geq  \frac{1}{8}.
\end{equation*}
\par
Let us now prove (\ref{M_2}).
We have
\begin{equation*}
    M_2\varphi=  A^{-1}(\theta-\tilde{q}_k\alpha)A^{-1}(\theta+\alpha-\tilde{q}_k\alpha)
      \cdots    A^{-1}(\theta+(j_0-1)\alpha-\tilde{q}_k\alpha)  \varepsilon_{j_0} \varphi(j_0-\tilde{q}_k+1),
 \end{equation*}
 and
 \begin{equation*}
    \left(\begin{array}{c}
                                                                                            u_{j_0-\tilde{q}_k} \\
                                                                                            u_{j_0-\tilde{q}_k-1}
                                                                                          \end{array}\right)=
                                                                                          \left(
                                                                                            \begin{array}{cc}
                                                                                              0 & 1 \\
                                                                                              -1 &  e-\lambda \tan\pi(\theta+j_0\alpha -\tilde{q}_k\alpha) \\
                                                                                            \end{array}
                                                                                          \right)
                                                                                          \left(\begin{array}{c}
                                                                                            u_{j_0-\tilde{q}_k+1} \\
                                                                                            u_{j_0-\tilde{q}_k }
                                                                                          \end{array}\right) .
 \end{equation*}
Since  $\{u_n\}\in \ell^2(\mathbb{Z})$,
 one has
 \begin{equation}\label{varepsilon_j_0111}
     (e-\lambda \tan\pi(\theta+j_0\alpha -\tilde{q}_k\alpha)) u_{j_0-\tilde{q}_k }=O(1) .
 \end{equation}
 Clearly,
 \begin{equation}\label{varepsilon_j_01}
   || \varepsilon_{j_0} \varphi(j_0-\tilde{q}_k+1)||=O(\frac{u_{j_0-\tilde{q}_k }}{|c_{j_0}^2|\tilde{q}_{k+1} } ).
 \end{equation}
 If $ c_{j_0}$ is small enough, i.e., $ |c_{j_0}|\leq  \frac{1}{C} $,  then
 we have
 \begin{equation*}
   | c_{j_0} (e-\lambda \tan\pi(\theta+j_0\alpha -\tilde{q}_k\alpha))|\geq \frac{1}{C}.
 \end{equation*}
Combining with   (\ref{varepsilon_j_0111}) and (\ref{varepsilon_j_01}),  one has
 \begin{equation}\label{varepsilon_j_02}
   || \varepsilon_{j_0} \varphi(j_0-\tilde{q}_k+1)||\leq \frac{C}{|c_{j_0} |\tilde{q}_{k+1} }.
 \end{equation}
 If  $ |c_{j_0}|\geq  \frac{1}{C} $, then(\ref{varepsilon_j_01}) also implies (\ref{varepsilon_j_02}).
 \par
 It is easy to see that (Appendix \ref{appendixLemB_1})
 \begin{equation*}
    ||A^{-1}(\theta-\tilde{q}_k\alpha)A^{-1}(\theta+\alpha-\tilde{q}_k\alpha)
      \cdots    A^{-1}(\theta+(j_0-1)\alpha-\tilde{q}_k\alpha)||\leq Ce^{\tilde{q}_k(\gamma_\lambda(e)+\varepsilon)}.
 \end{equation*}
 Thus by (\ref{Def.c_0}), one has
 \begin{eqnarray*}
    || M_2\varphi|| &\leq& e^{\tilde{q}_k(\gamma_\lambda(e)+\varepsilon)}\frac{C}{|c_{j_0} |\tilde{q}_{k+1} }  \\
     &\leq& Ce^{-(\delta(\alpha,\theta)- \gamma_\lambda(e)-2\varepsilon)\tilde{q}_k   }
 \end{eqnarray*}
  This implies (\ref{M_2}).
  \par
 Case 3: if  $||A_{\tilde{q}_k}(\theta)\varphi||< \frac{1}{4}$ and
 $||A_{\tilde{q}_k}^2(\theta)\varphi||\geq \frac{1}{4}$, using
 (\ref{Es1a}) and similarly as in the proof of case 2, we also have
 \begin{eqnarray}
 \nonumber
   ||A_{2\tilde{q}_k}(\theta)\varphi- A_{ \tilde{q}_k}^2(\theta)\varphi|| &=& ||(A_{\tilde{q}_k}(\theta+\tilde{q}_k\alpha)-A_{\tilde{q}_k}(\theta) )
    A_{ \tilde{q}_k} (\theta)\varphi ||\\
    \nonumber
    &\leq&  Ce^{-(\delta(\alpha,\theta)- \gamma_\lambda(e)-\varepsilon)\tilde{q}_k   }\leq1/8,
 \end{eqnarray}
This   yields that $||\varphi(\tilde{q}_{2k})||\geq\frac{1}{8}$.

Therefore, we have that, in either case,  $\max
\{\varphi(\tilde{q}_{k}),\varphi(\tilde{-q}_{k}),\varphi(\tilde{q}_{2k})\}\geq
\frac{1}{8}$ which contradicts the  fact $\{u_n\}\in \ell^2(\mathbb{Z})$.
\par
Now we are in a position to prove the second  part of Theorem \ref{maintheorem1}.
By  the definition of $ \beta(\alpha) $ (\ref{Def.beta}),  there  exists a subsequence $\hat{q}_k$ of $q_n$ such that,
  \begin{equation*}
  \lim_{k\rightarrow\infty}\frac{\ln \hat{q}_{k+1}}{\hat{q}_k} =\beta(\alpha).
 \end{equation*}
 Thus in order to prove the second part, it suffices to prove that for almost every $\theta$,
 \begin{equation*}
     \limsup_{k\rightarrow\infty}\frac{\ln||\hat{q}_k(\theta-1/2)||_{\mathbb{R}/\mathbb{Z}}}{\hat{q}_k}>  \gamma_\lambda(e)-\beta(\alpha) .
 \end{equation*}
 Suppose for some $\theta$,
  \begin{equation*}
     \limsup_{k\rightarrow\infty}\frac{\ln||\hat{q}_k(\theta-1/2)||_{\mathbb{R}/\mathbb{Z}}}{\hat{q}_k}\leq \gamma_\lambda(e)-\beta(\alpha) .
 \end{equation*}
 This implies $\beta(\theta-1/2)\geq \beta(\alpha)-\gamma_\lambda(e)  >0$.
  Notice however that  $ \{\theta\in \mathbb{R}: \beta(\theta-1/2)>0\}$ is a measure zero set.
$\qed$
\section{The pure point spectrum}
We will first prove the following
\begin{theorem}\label{maintheorem2}
With $\delta (\alpha, \theta)$ as in (\ref{Def.delta_1}),
$h_{\lambda,\alpha,\theta}$  has only pure  point spectrum
 (and the corresponding
eigenfunctions are
  exponentially decaying)
 on  $\{e:\gamma_{\lambda}(e) >\delta (\alpha, \theta) \}$.
\end{theorem}

 Following \cite{MR776654} we will use the special algebraic structure
 of the Maryland model.
First, we use Cayley transform and Fourier transform to give  another form  of
the eigenvalue problem.
\par
Introduce two new operators, 
\begin{equation}
\nonumber
    B_1=\lambda^{-1}(e-\Delta), B_2=\text{mult. by } \tan\pi(\theta+n\alpha),
\end{equation}
where $(\Delta u)_n=u_{n+1}+u_{n-1}$.
\par

Let
\begin{equation*}
    \mathcal{P}=\{\{u_n\}_{n\in \mathbb{Z}}: |u_n|\leq C(1+|n|^{C} )\text{ for some constant }C \}.
\end{equation*}

\begin{lemma}(\cite{cyconschrodinger,MR776654})\label{TRF1}
If   a vector $u\in  \mathcal{P} $ ($u\in\ell^2(\mathbb{Z})$ or is exponentially decaying)  satisfies  $B_2u=B_1u$, then   $c=(1+iB_1)u\in  \mathcal{P} $ ($c\in\ell^2(\mathbb{Z})$ or is exponentially decaying) and
\begin{equation}\label{Clay}
    (1-iB_2)(1+iB_2)^{-1}c=(1-iB_1))(1+iB_1)^{-1}c.
\end{equation}
The converse is also true. More precisely, if $c\in  \mathcal{P} $ ($c\in\ell^2(\mathbb{Z})$ or is exponentially decaying) and satisfies (\ref{Clay}), then
the vector $ u=(1+iB_1)^{-1}c \in \mathcal{P}$ ($u\in\ell^2(\mathbb{Z})$ or is exponentially decaying) and $B_1u=B_2u$.
\end{lemma}
It is easy to see that operator $ (1-iB_2)(1+iB_2)^{-1}$ (this is called Cayley transform of $B_2$) is multiplication by
\begin{equation*}
    e^{-2\pi i(n\alpha+\theta) }
\end{equation*}
and (\ref{Clay}) becomes
\begin{equation}\label{TRF2}
    \left(\frac{1-iB_1}{1+iB_1 }  c\right)_n=  e^{-2\pi i(n\alpha+\theta) }c_n.
\end{equation}
Define a Fourier Transform (in distributional sense) of a sequence $f_n \in \mathcal{P}$ by
\begin{equation}
\nonumber
    \hat{f}(x)=\sum_n e^{-2\pi inx}f_n.
\end{equation}

Then equation (\ref{TRF2}) becomes
\begin{equation}\label{TRF3}
    q(x)\hat{c}(x)=e^{-2\pi i\theta}\hat{c}(x+  \alpha),
\end{equation}
where $q(x)=-\frac{2\cos2\pi x-e-i\lambda}{2\cos 2\pi x-e+i\lambda}$.
\begin{lemma}(\cite{MR776654})\label{LemValuezeta}
Let $q(x)=e^{-2\pi i\zeta(x)}$ with $-1/2<\zeta(x)<1/2 $, then
\begin{equation*}
    \zeta(x)=\sum_n\zeta_ne^{-2\pi inx },
\end{equation*}
where \begin{eqnarray}
         \zeta _0&=&  k_{\lambda}(e)-\frac{1}{2}  , \label{Valuezeta_0}\\
          \zeta _n &=&  (-1)^n\frac{1}{n\pi}e^{-\gamma_{\lambda}(e)|n|}\sin(\pi n k_{\gamma}(e)), n\neq 0,\label{Valuezeta_n}
      \end{eqnarray}
   where $k_{\lambda}(e)$ is the IDS.
\end{lemma}
First, assume $e$ satisfies the {\it quantization condition } (Proposition
\ref{Proquantizationcondition}).  Let us show that
$e$ is an eigenvalue   with  exponentially decaying eigenfunction in the regime $\{e:\gamma_{\lambda}(e) >\delta (\alpha, \theta)$\}.
\begin{lemma}\label{Firstpart}
If $e$ satisfies $k_\lambda(e)\in \theta-1/2+ \alpha\mathbb{Z}+\mathbb{Z}$ and $\gamma_{\lambda}(e) >\delta (\alpha, \theta)$, then $e$ is an eigenvalue of operator $h_{\lambda,\alpha,\theta}$ and the corresponding
eigenfunction is
  exponentially decaying.
\end{lemma}
\begin{proof}
It suffices to prove that  under the condition of Lemma \ref{Firstpart},  the following equation
 \begin{equation}\label{Mainequation}
   e^{-2\pi i\zeta(x)} \hat{c}(x)=e^{-2\pi i\theta}\hat{c}(x+  \alpha),
\end{equation}
has an analytic (on $\mathbb{T}$) solution. Indeed, if $\hat{c}(x)$ is analytic on $\mathbb{T}$, then  one has
\begin{equation*}
   | c_n|\leq C e^{-t |n|},
\end{equation*}
for some $t>0$, where $c_n$ is the Fourier coefficient of $ \hat{c}(x)$.  This implies  $c\in \ell^2(\mathbb{Z})$ is
 exponentially decaying. This also implies   $e$ is an eigenvalue and $u\in \ell^2(\mathbb{Z})$ is
 exponentially decaying by Lemma \ref{TRF1}.
 \par
 Assume $k_\lambda(e)=\theta-1/2+ m \alpha $ for some $m\in\mathbb{Z}$.
 We will prove that equation (\ref{Mainequation}) has an  analytic solution  of the form $\hat{c}(x)=e^{-2\pi i(mx+\psi(x))}$, where $\psi(x)$ is real analytic on $\mathbb{T}$.
  Substituting $\hat{c}(x)=e^{-2\pi i(mx+\psi(x))}$ into   (\ref{Mainequation}) and comparing the Fourier coefficients,
one has
 \begin{equation}\label{Harmonicequation}
    \psi_k=\frac{\zeta _k}{ e^{- 2\pi  i k\alpha} -1}, k\neq 0.
 \end{equation}
 Thus in order to prove this lemma, we only need to prove that $\psi_k$
 is exponentially decaying.
 \par
 By the definition of $\delta (\alpha,\theta)$
, for any $\varepsilon>0$, we have
 \begin{equation}\label{q_ntheta}
    ||q_n(\theta-1/2)||_{\mathbb{R}/\mathbb{Z}}\leq \frac{e^{(\delta  +\varepsilon)q_n}}{q_{n+1}},
 \end{equation}
 for $n$ large enough. Below, we always suppose $n$ is sufficient large.
 \par
 Thus
 \begin{equation*}
    ||q_nk_{\gamma}(e)||_{\mathbb{R}/\mathbb{Z}}\leq \frac{e^{(\delta +\varepsilon)q_n}}{q_{n+1}}+|m|\;||q_n\alpha||_{\mathbb{R}/\mathbb{Z}}.
 \end{equation*}
 \par
 If $|k|=q_n$, by (\ref{Valuezeta_n}) one has
 \begin{equation*}
    | \zeta_{k}  | \leq  Ce^{-\gamma_{\lambda}(e)q_n}|| q_n k_{\gamma}(e))||_{\mathbb{R}/\mathbb{Z}}.
 \end{equation*}
 Then
 \begin{eqnarray*}
   |\psi_{k}| &\leq& C\frac{e^{-\gamma_{\lambda}(e)q_n}|| q_n k_{\gamma}(e))||_{\mathbb{R}/\mathbb{Z}}}{|| q_n \alpha||_{\mathbb{R}/\mathbb{Z}}}  \\
     &\leq&  C e^{-(\gamma_{\lambda}(e)-\delta -\varepsilon)q_n}  +C|m|   e^{-\gamma_{\lambda}(e)q_n} .
 \end{eqnarray*}
 In this case,   $ |\psi_{k}|$ is exponentially decaying because   $ \gamma_{\lambda}(e) >\delta (\alpha, \theta)  $.
 \par
 If $q_n<|k|<q_{n+1}$ and $|k|\geq 2\frac{\ln q_{n+1}}{\gamma_\lambda(e)}$, by (\ref{Harmonicequation}) one has
 \begin{eqnarray*}
   |\psi_{k}| &\leq & C \frac{e^{-\gamma_{\lambda}(e)|k|}}{|| q_n \alpha||_{\mathbb{R}/\mathbb{Z}}} \\
     &\leq& C  e^{-\frac{\gamma_{\lambda}(e)}{2}|k|}.
 \end{eqnarray*}
  \par
 If $q_n<|k|<q_{n+1}$ and $|k|\leq 2\frac{\ln q_{n+1}}{\gamma_\lambda(e)}$, let $|k|=\ell q_n+k_0$ with $|\ell|\leq2\frac{\ln q_{n+1}}{q_n\gamma_\lambda(e)} $
 and $0\leq k_0<q_n$.
 Assume $k_0\neq 0$, then by (\ref{Approximate1}) and (\ref{Approximate2}), we have
   \begin{eqnarray*}
                                          ||k\alpha||_{\mathbb{R}/\mathbb{Z}} &\geq&  \Delta_{n-1}-\ell \Delta_{n} \\
                                           &\geq& \frac{1}{C q_n},
                                       \end{eqnarray*}
  Clearly,  $ \psi_{k}$ is exponentially decaying.
  \par
  Assume $k_0=0$, then  $|k|=\ell q_n $ with $|\ell|\leq2\frac{\ln q_{n+1}}{q_n\gamma_\lambda(e)} $. Applying the
  same proof as in the  first case,
  we also have $  \psi_{k} $ is exponentially decaying.
\end{proof}
\par

It is well known that in order to  prove Theorem \ref{maintheorem2},  it suffices to prove that any polynomially  bounded solution
  (i.e., $u\in\mathcal{P}$) of (\ref{MME}) is  exponentially decaying.
  By  the Cayley  transform (Lemma \ref{TRF1})  and Fourier transform, we only need to prove that
   if the   equation (\ref{Mainequation})
 has a solution $\hat{c}(x)$ such that $\hat{c}(x)$ is a Fourier transform of a sequence $c_n\in \mathcal{P}$, then
 $c_n$ is  exponentially decaying (or $\hat{c}(x)$ is analytic on $\mathbb{T}$).
Applying Lemma \ref{Firstpart},  it suffices to show that
   $ k_\lambda(e)\in \theta-1/2+ \alpha\mathbb{Z}+\mathbb{Z}$.
The next two results will establish this fact.
\par
\begin{lemma}\label{Twocondition}
Suppose
the  equation
 $
   e^{-2\pi i\zeta(x)} \hat{c}(x)=e^{-2\pi i\theta}\hat{c}(x+  \alpha),
$
 has a solution $\hat{c}(x)$ such that $\hat{c}(x)$ is a Fourier transform of a sequence $c_n\in \mathcal{P}$, and there exists
 a sequence $t_k\in \mathbb{N}^+$ and $t>0$ that satisfying the following two conditions:
\par
(1)$\lim_{k\rightarrow\infty}||t_k\alpha||_{\mathbb{R}/\mathbb{Z}}=0$;
\par
(2)$\sum_{j\neq 0}\sup_{k}|\frac{1-e^{-2\pi ijt_k\alpha}}{1-e^{-2\pi ij\alpha}}  {\zeta} _j |e^{ t|j|}<\infty$.
\par
Then we have
\begin{equation*}
    \lim_{k\rightarrow\infty}||t_k(k_{\lambda}(e)-1/2-\theta)||_{\mathbb{R}/\mathbb{Z}}=0.
\end{equation*}
\end{lemma}
\begin{proof}
Clearly, one has
 \begin{equation}\label{new1}
   e^{-2\pi i\sum_{j=0}^{t_k-1}\zeta(x+j\alpha)}\hat{c}(x)=e^{-2\pi it_k\theta}\hat{c}(x+ t_k  \alpha).
\end{equation}
The first condition implies that $ \lim_{k\rightarrow\infty}\hat{c}(x+ t_k  \alpha)=\hat{c}(x)$ in the distributional sense.
 The second condition implies that
 \begin{equation*}
    \lim_{k\rightarrow\infty}\sum_{j=0}^{t_k-1}\zeta(x+j\alpha)- t_k {\zeta} _0 =0,
\end{equation*}
for $x$ in the  strip $ \{z:|\Im z|\leq t/2\},$ uniformly.
\par
Combining with (\ref{new1}), we have
\begin{equation*}
    \lim_{k\rightarrow\infty}||t_k(  {\zeta} _0-\theta)||_{\mathbb{R}/\mathbb{Z}}=0.
\end{equation*}
Using (\ref{Valuezeta_0}), we obtain this lemma.
\end{proof}
\par
\begin{proposition}\label{Proquantizationcondition}
If  $\gamma_{\lambda}(e) >\delta (\alpha, \theta)$ and   equation (\ref{MME}) has a solution $u\in \mathcal{P}$,
then   $e$  satisfies
\begin{equation}\label{quantizationcondition}
     k_{\lambda}(e)\in \theta+1/2+\alpha\mathbb{Z}+\mathbb{Z} .
\end{equation}
\end{proposition}
\begin{proof}
Let
\begin{equation*}
\{t_k\}_{k=1}^{\infty}=\{q_n, 2q_n, 3q_n,\cdots, \ell_{ n} q_n\} _{n=1}^{\infty} ,
\end{equation*}
   with $   \ell_{ n} = \min\{ \lfloor\frac{ q_n^2}{||q_nk_{\lambda}(e)||_{\mathbb{R}/\mathbb{Z}}}\rfloor, \lfloor \frac{2q_{n+1} }{q_n}  \rfloor\} $, where $\lfloor m\rfloor$ denotes the smallest integer not exceeding $m$.
   \par
   By the fact that
 $  \ell_{ n}\leq \frac{2q_{n+1} }{q_n} $, one has  that $t_k$ satisfies the first  condition of Lemma \ref{Twocondition}.
Combining with that  $\ell_{ n}  \leq  \frac{
  q_n^2}{||q_nk_{\lambda}(e)||_{\mathbb{R}/\mathbb{Z}}}$ and using an
argument similar to the  proof of Lemma \ref{Firstpart}, we have that $t_k$ satisfies the second condition of
Lemma \ref{Twocondition} (see Appendix \ref{A.1}).
 \par
 By (\ref{TRF3}), one has that the  equation
 $
   e^{-2\pi i\zeta(x)} \hat{c}(x)=e^{-2\pi i\theta}\hat{c}(x+  \alpha),
$
 has a solution $\hat{c}(x)$ such that $\hat{c}(x)$ is a Fourier transform of a sequence $c=(I+iB_1)u\in \mathcal{P}$.
 \par
 Applying Lemma \ref{Twocondition},
 we have
\begin{equation*}
    \lim_{k\rightarrow\infty}||t_k(k_{\lambda}(e)-\frac{1}{2}-\theta)||_{\mathbb{R}/\mathbb{Z}}=0.
\end{equation*}
This implies that
\begin{eqnarray*}
  ||q_n(k_{\lambda}(e)-\frac{1}{2}-\theta)||_{\mathbb{R}/\mathbb{Z}} &\leq&  \frac{1}{\ell_n} \\
    &\leq&  \max \{C  \frac{||q_n k_{\lambda}(e)||_{\mathbb{R}/\mathbb{Z}} }{q_n^2}, C\frac{q_n}{q_{n+1} } \}.
\end{eqnarray*}
Combining with   (\ref{q_ntheta}), we have
 \begin{equation*}
    ||q_n k_{\lambda}(e)  ||_{\mathbb{R}/\mathbb{Z}} \leq\max \{C  \frac{||q_nk_{\lambda}(e)||_{\mathbb{R}/\mathbb{Z}} }{q_n^2}+\frac{e^{(\delta +\varepsilon)q_n}}{q_{n+1}}, C\frac{q_n}{q_{n+1} }+\frac{e^{(\delta +\varepsilon)q_n}}{q_{n+1}}\}.
\end{equation*}
In particular,
 \begin{equation*}
    ||q_n k_{\lambda}(e)  ||_{\mathbb{R}/\mathbb{Z}} \leq C\frac{q_n}{q_{n+1} }+C\frac{e^{(\delta +\varepsilon)q_n}}{q_{n+1}} .
\end{equation*}

\par
By an argument similar to that in the proof of Lemma  \ref{Firstpart}, we obtain that the following  equation
\begin{equation*}
   e^{-2\pi i\tilde{\zeta}(x)} \hat{c}_1(x)= \hat{c}_1(x+  \alpha),
\end{equation*}
has an analytic solution $\hat{c}_1(x)$, where $\tilde{\zeta}(x)=\zeta(x)-\zeta_0 $ (Appendix \ref{A.2}).
\par
Let $d(x)=\frac{\hat{c}(x)}{\hat{c}_1(x)}$ (well defined in the
distributional sense since $\hat{c}_1(x)$ is analytic). It obeys
\begin{equation*}
    d(x+  \alpha )=e^{-2\pi i(k_{\lambda}(e)-\frac{1}{2}-\theta)}d(x).
\end{equation*}
This implies that
\begin{equation*}
    d(x+  j_n\alpha )=e^{-2\pi ij_n(k_{\lambda}(e)-\frac{1}{2}-\theta)}d(x).
\end{equation*}
Thus,
 for any sequence $j_n\in \mathbb{Z}$ such that $ \lim_{n\rightarrow\infty}||j_n \alpha  ||_{\mathbb{R}/\mathbb{Z}}=0$,
 we must have
 \begin{equation*}
  \lim_{n\rightarrow\infty}||j_n (k_{\lambda}(e)-\frac{1}{2}-\theta) ||_{\mathbb{R}/\mathbb{Z}}=0.
 \end{equation*}
 This implies $k_\lambda (e)-\frac{1}{2}-\theta=m\alpha \mod \mathbb{Z}$ for some $m\in \mathbb{Z}$.
 \end{proof} 
 \par
From above discussion  one has  that  the eigenfunctions
 corresponding to the eigenvalues  in the regime  $\{e:\gamma_{\lambda}(e) >\delta (\alpha, \theta)  \}$ are exponentially decaying.
  Although the (up to) two points $\{e:\gamma_{\lambda}(e)=\delta
  (\alpha, \theta)  \}$ may or may not be in the point spectrum, we can
  prove that
  the  corresponding eigenfunctions (if they exist) 
  can not be exponentially decaying.
  \begin{proposition}\label{exponentiallydecayingeigenfunctions}
  If $e\in \sigma_{p}(h_{\lambda,\alpha,\theta})$ and the  corresponding eigenfunction is exponentially decaying, then
  we must have $\gamma_{\lambda}(e) >\delta (\alpha, \theta)  $.
  \end{proposition}
  \begin{proof}
 Assume  $e$ is an eigenvalue with an  exponentially decaying  eigenfunction. Then
   equation (\ref{Mainequation}) must have an analytic solution
   $\hat{c}(x)$. By ergodicity, this implies $|\hat{c}(x)|$ is
   constant, so without loss of generality, assume
   $|\hat{c}(x)|=1$. Thus $\hat{c}(x)$ defines  a map from
   $\mathbb{T}$ to the circle, and we may assume
   \begin{equation}\label{winding}
  \hat{c}(x)=e^{-2\pi i(mx+\psi(x))},
   \end{equation}
   where $m$ is the winding number of map $\hat{c}(x)$ and $\psi(x)$ is real analytic on $\mathbb{T}$.
   \par
Combining (\ref{winding}) with  (\ref{Mainequation}) and comparing the Fourier coefficients,
one has
\begin{equation}\label{winding1}
k_\lambda(e)=\theta-1/2+ m\alpha \mod \mathbb{Z}
\end{equation}
and

 \begin{equation}\label{Harmonicequation1}
    \psi_k=\frac{\zeta _k}{ e^{- 2\pi  i k\alpha} -1}, k\neq 0.
 \end{equation}

  Assume
  \begin{equation*}
   |\psi_k|\leq C e^{-2t|k|}.
  \end{equation*}

  In (\ref{Harmonicequation1}), let $k=q_n$, then combining with (\ref{Valuezeta_n})
   one has
   \begin{equation*}
    ||q_nk_{\lambda}(e)||_{\mathbb{R} /\mathbb{Z}}\leq  C\frac{e^{(\gamma -t)q_n} }{q_{n+1}},
 \end{equation*}
 for some $t>0$.
 \par
 Using (\ref{winding1}), we have
 \begin{equation*}
    ||q_n (\theta-1/2)||_{\mathbb{R} /\mathbb{Z}}\leq  C\frac{e^{(\gamma -t)q_n} }{q_{n+1}}.
 \end{equation*}
 This implies $ \delta (\alpha, \theta) \leq \gamma_{\lambda}(e)  -t< \gamma_{\lambda}(e) $.
 \end{proof}
We are now ready to prove Theorem \ref{quant}

\begin{proof}

Set $B_{\lambda,\alpha,\theta}=\{e:
||t_k(k_\lambda(e)-\theta-1/2)||\to 0 \;\mbox{as}\;
k\to\infty\}$, where the sequence $t_k$ is defined in lemma \ref{A.1}. Clearly, $B_{\lambda,\alpha,\theta}$ is a proper
subset of $A_{\lambda,\alpha,\theta}$.    Combining Lemma
\ref{Twocondition} and Lemma \ref{A.1} we obtain part (1).

By Theorem \ref{maintheorem1}, part (2) is a combination of Lemma \ref{Firstpart},
Proposition \ref{Proquantizationcondition} and Proposition \ref{exponentiallydecayingeigenfunctions}.

Part (3) follows immediately from Theorem \ref{maintheorem1} and  part (2).
\end{proof}
\section{ Constancy of the spectrum. Proof of Theorems \ref{Maintheorem} and \ref{Spectrumforalltheta}}
In this section, we only prove a simple fact that the spectrum of
Maryland model $ \sigma(h_{\lambda,\alpha,\theta})$ does not depend on
parameter $\theta$ which will prove Theorem
\ref{Spectrumforalltheta}. This will aloow us to complete the proof of
Theorem \ref{Maintheorem}.

Following the discussion of unbounded operators in section VIII
\cite{reed1972methods},  it suffices to show that self-adjoint operators  $h_{\lambda,\alpha,\theta}$
are continuous in norm resolvent sense with respect to $\theta$.  This
can be done using a specific form of  the Green function of Maryland
model given in  (p.349, \cite{carmona1990spectral}).
Here are the details.

{\bf Proof of Theorem \ref{Spectrumforalltheta}}

Let $U$ be the operator $(U\psi)_n=e^{2\pi i n\alpha}\psi_n$. Then we have the following
\begin{lemma}(p.349, \cite{carmona1990spectral})
Let  $ G_z(\theta)=(h_{\lambda,\alpha,\theta}-zI)^{-1}$   and $G_{0,z }=(h^{0}-zI)^{-1}$,
where $h^{0} $ is the free  Schr\"{o}dinger operator (i.e., $h^{0}= h_{0,\alpha,\theta} $).
Then one has
\begin{equation} \label{Def.Green}
     G_z(\theta)=(I+e^{2\pi i \theta}U)(I+e^{2\pi i\theta}CU)^{-1}G_{0,z-i\lambda}
\end{equation}
for all $\lambda \cdot \Im z<0$, where $C=(h^0-(z+i\lambda)I)G_{0,z-i\lambda }$.
\end{lemma}
Notice that $ ||C||<1$ for all  $\lambda \cdot \Im z<0$, so that
(\ref{Def.Green}) is well defined.

Given  two phases $\theta_1$ and $\theta_2$, by  ergodicity,   there exists a sequence $j_k$ of integers
such that
\begin{equation*}
 \lim_{k\rightarrow \infty}||\theta_1 +j_k \alpha-\theta_2||_{\mathbb{R}/\mathbb{Z}}\rightarrow 0 .
\end{equation*}
 Let $B_k= h_{\lambda,\alpha,\theta_1+j_k \alpha}$ and
$B=h_{\lambda,\alpha,\theta_2}$. We have that $B_k$   converges to $B$ in norm resolvent sense.
Indeed, if $ \Im z<0$, this is easy to see by (\ref{Def.Green})
(recall that we always assume $\lambda>0$ in this paper); if $ \Im
z>0$,  this is also true  by the fact that  $h_{\lambda,\alpha,\theta}=h_{-\lambda,-\alpha,-\theta}$.
   \par
  Thus, according to Theorem VIII.24, \cite{reed1972methods} on norm
  resolvent convergence, we have that for any $e\in \sigma(h_{\lambda,\alpha,\theta_2})$, there exists
  $e_k\in \sigma(B_k)$ such that $e_k\rightarrow e$.  Clearly, the spectrum of operator  $B_k$ does not depend on $k$, i.e.,
  $\sigma(B_k)=\sigma(h_{\lambda,\alpha,\theta_1})$. Thus
   we must have
   \begin{equation*}
     \sigma(h_{\lambda,\alpha,\theta_2})\subseteq  \sigma(h_{\lambda,\alpha,\theta_1}).
   \end{equation*}
This implies the  theorem since $ \theta_1$ and $ \theta_2$ were arbitrary.

\textbf{Proof of Theorem \ref{Maintheorem}}:
\par
By Theorem \ref{maintheorem1} and \ref{maintheorem2}, we obtain that
$h_{\lambda,\alpha,\theta}$ has purely singular continuous spectrum
in the regime   $\{e:\gamma_{\lambda}(e) <\delta (\alpha, \theta) \}$
and only pure point in the regime   $\{e:\gamma_{\lambda}(e) >\delta (\alpha, \theta) \}$.
 Then Theorem \ref{Maintheorem} follows from
 Theorem \ref{Spectrumforalltheta}   and  a simple fact that
 $\sigma_{pp}(h)$ and $\sigma_{sc}(h)$ are closed sets. Notice that we
 have used  the fact that $\sigma_{sc}(h)$ can not have isolated  points (corresponding to Fig.2).
\par
Next we show that for any $\alpha,$ there exist phases with Anderson localization.
\begin{corollary}\label{cannotextendtoalltheta}
 Let
 \begin{equation*}
    S=\{\theta\in
    \mathbb{R}:||q_{n}(\theta-1/2)||_{\mathbb{R}/\mathbb{Z}} <10\frac{q_n}{q_{n+1}}\text{ eventually in } n\},
 \end{equation*}
 then
   $S $ is a dense uncountable set  and
 $h_{\lambda,\alpha,\theta}$ satisfies Anderson Localization
for every $\theta \in S$.  Moreover, the point spectrum is determined  by the quantization condition (\ref{quantizationcondition}).
\end{corollary}
\begin{proof}
Consider the set $S_1$
\begin{eqnarray*}
   S_1&=&\{\theta\in \mathbb{R}: \text{ eventually in } n, \text{ there exists some } \\
 &\;&0\leq j\leq q_n-1 \text{ such that }|| \theta-1/2+\frac{j}{q_n} ||_{\mathbb{R}/\mathbb{Z}}
    <\frac{10}{q_{n+1}}\}.
\end{eqnarray*}
 It is clear that $S_1$ is a dense uncountable set and $S_1\subseteq
 S$. This yields that $S $ is also uncountable and dense.
 \par
 Moreover, for $\theta\in S$, one has $\delta
 (\alpha,\theta)\leq0$. Using the fact that $\gamma_{\lambda}(e)>0$, Theorem
 \ref{maintheorem2} and Proposition \ref{Proquantizationcondition}, we obtain the corollary.
\end{proof}

\section{ Complexified Cocyles}
Here we show a simple way to compute $\gamma_\lambda(e)$ using ideas
of Avila's global theory.

Denote by $  C^{\omega}(\mathbb{T},M_2(\mathbb{C}))$ the class of $1-$periodic functions on $\mathbb{R}$,
with analytic extension to some strip, $|\Im z|<\delta$, attaining values in complex $2\times2$ matrices.

Cocycles $(\alpha,D)$ with $D\in
C^{\omega}(\mathbb{T},M_2(\mathbb{C})) $ are called analytic. For
analytic cocycles, the Lyapunov exponent $
L(\alpha,D):\mathbb{T}\times C^{\omega}(\mathbb{T},M_2(\mathbb{C}))\to
[-\infty,\infty)$ is jointly
continuous at every $(\alpha,D)$ with $\alpha\in  \mathbb{R}\backslash
\mathbb{Q}$ \cite{jitomirskaya2012analytic,bourgain2002continuity,MR2121278}.  Given any analytic  cocyle $(\alpha,D)$,
we consider its holomorphic extension, $(\alpha,D(x+i\epsilon))$  with  $|\epsilon|\leq\delta$ for some appropriate $\delta>0$.
$L(\alpha,D_\epsilon)$ is referred to the Lyapunov exponent of the
complexified cocyle $(\alpha,D(x+i\epsilon))$, so $L(\alpha,D)=L(\alpha,D_0)$.
\par
The   Lyapunov exponent $L(\alpha,D_\epsilon)$ is easily seen to be a
convex function of $\epsilon$.
Thus we can introduce the acceleration of
$(\alpha,D)$,
\begin{equation*}
      \omega(\alpha,D;\epsilon)=\frac{1}{2\pi}\lim_{h\rightarrow0+}\frac{L(\alpha,D_{\epsilon+h})-L(\alpha,D_\epsilon)}{h}.
\end{equation*}
It follows from convexity and continuity of the Lyapunov exponent that the acceleration is an upper semi-continuous function in parameter $\epsilon.$
The acceleration was first introduced and
the above results proved in \cite{avila2009global} for analytic
$SL(2,\mathbb{C})$-cocycles. It was extended to general case $M_2(\mathbb{C})$ in \cite{jitomirskaya2012analytic}.
\par
The acceleration satisfies
\begin{theorem}(Quantization of acceleration,\cite{MR3010376,avila2009global,avila2013complex})\label{QAT}
Consider cocycle $(\alpha,D)$ with $\det D(x)$  is bounded away from zero on the strip $\mathbb{T}_{\delta}=\{z:|\Im z|<\delta\}$, then
$ \omega(\alpha,D;\epsilon)\in \frac{1}{2}\mathbb{Z}$.
\end{theorem}
Moreover, $ \omega(\alpha,D;\epsilon)\in  \mathbb{Z}$ for $\text{SL}(2,\mathbb{C})$-cocycles \cite{avila2009global,avila2013complex}.
\par
After the above preparation, we can  compute the Lyapunove exponent of
the Maryland model.
\par
Recall that the Maryland cocycle $(\alpha,A)$ is given by
\begin{equation*}\label{G24}
   A(x)=\left(
          \begin{array}{cc}
            e-\lambda\tan\pi x &  -1 \\
            1 & 0 \\
          \end{array}
        \right).
 \end{equation*}
  \begin{lemma} \label{Lecontinuousepsilon}
    $L(\alpha, A_\epsilon)$ is  a continuous function with respect to parameter $\epsilon$.
  \end{lemma}
\begin{proof} Recall that by (\ref{MMR}),  $D(x)=\cos \pi x A(x)$ is
  an analytic cocycle, thus $L(\alpha,D_\epsilon)$ is continuous in $\epsilon.$

Set $I_{\epsilon}=\int_{\mathbb{T}}\ln|\cos\pi(x+i\epsilon)|dx$.
We have
                       \begin{equation*}
                        L(\alpha,A_\epsilon)= L(\alpha,D_\epsilon)-I_\epsilon.
                        \end{equation*}
 Applying Jensen's formula yields (\cite{MR2121278} or \cite{jitomirskaya2012analytic})
 \begin{equation}\label{CosV}
     I_\epsilon= \pi |\epsilon|-\ln 2,  \text{ for all } \epsilon.
 \end{equation}

Since  (\ref{CosV}) explicitly implies  the continuity of $I_\epsilon
$ in  $\epsilon$, the continuity of $L(\alpha,A_\epsilon)$ follows.
 \end{proof}

 \begin{theorem}\label{Thmvaluele}
 Let $\gamma_\lambda(e)=L(\alpha,A)  $, then   $\gamma_\lambda(e)>0$ and    the following equation holds,
 \begin{equation}\label{LEBiaoda1}
     4\cosh \gamma_\lambda(e)=  \sqrt{ (2+e)^2 +\lambda^2} +\sqrt{ (2-e)^2 +\lambda^2}.
 \end{equation}

 \end{theorem}
 \begin{proof}
 First, uniformly in $x\in \mathbb{T}$,  one has
 \begin{equation*}
        A(x+i\epsilon)= A_{\infty}+o(1),
 \end{equation*}
 as $\epsilon\rightarrow +\infty$,
 where
 \begin{equation*}
    A_{\infty}=\left(
                         \begin{array}{cc}
                           e-i\lambda &-1 \\
                           1 &0 \\
                         \end{array}
                       \right).
 \end{equation*}
 By  continuity of the LE,
 \begin{equation*}
    L(\alpha,A_\epsilon)=L(\alpha,A_\infty)+o(1),
 \end{equation*}
  as $\epsilon\rightarrow +\infty$.
  \par
  The quantization of acceleration yields
   \begin{equation*}
    L(\alpha,A_\epsilon)=L(\alpha,A_\infty), \text{ for all }   \epsilon>0  \text{ sufficiently large}.
  \end{equation*}

   In addition the convexity, continuity and symmetry of $L(\alpha,A_\epsilon)$ with respected to $\epsilon$, gives
  \begin{equation*}
    L(\alpha,A_\epsilon)=L(\alpha,A_\infty), \text{ for all }   \epsilon.
  \end{equation*}
  Note that symmetry means $L(\alpha,A_\epsilon)=L(\alpha,A_{-\epsilon})$.
\par
This implies
   \begin{equation*}
    \gamma_\lambda(e)=L(\alpha,A_\infty).
  \end{equation*}
Then Theorem \ref{Thmvaluele} follows from solving for the eigenvalue  of $A_{\infty}$ (a constant matrix) directly.
\end{proof}
\par

Finally, we note that for general Schr\"{o}dinger  operators
 \begin{equation*}
(Hu) _n=u_{n+1}+u_{n-1}+    v (  \theta+  n\alpha )u_{n},
 \end{equation*}
 where $v:\mathbb{T}\rightarrow \mathbb{R}$ is real analytic, Avila gives a complete description of spectrum under the condition of positive LE.
 \begin{theorem}(\cite{avila2009global})\label{Th24}
 For real analytic $v$, if LE $L(\alpha,B_e)>0$,  then $e\in \text{spec}(H)$ if and only if
  $ \omega(\alpha,B_e;0)>0$, where
   \begin{equation*}
   B_e(x)=\left(
          \begin{array}{cc}
            e-  v(  x) &  -1 \\
            1 & 0 \\
          \end{array}
        \right).
 \end{equation*}
\end{theorem}
For the Maryland model, however, we have $\gamma_{\lambda}(e)>0,$  $\text{spec}(h)=(-\infty,+\infty)$; yet
 $ \omega(\alpha,A_e;0)=0$ for all $e$.  This is similar to the
 phenomenon observed for the extended Harper's model for certain
 parameters (see Sec. 7 in \cite{jitomirskaya2012analytic}). Thus it gives
 another counterexample to Theorem \ref{Th24} in the singular case.
\appendix
\section{}
\begin{lemma}\label{appendixLemB_1}
The following estimate holds
\begin{equation*}
    ||A^{-1}(\theta-\tilde{q}_k\alpha)A^{-1}(\theta+\alpha-\tilde{q}_k\alpha)
      \cdots    A^{-1}(\theta+(j_0-1)\alpha-\tilde{q}_k\alpha)||\leq Ce^{\tilde{q}_k(\gamma_\lambda(e)+\varepsilon)}.
 \end{equation*}
 \end{lemma}
 \begin{proof}
 It suffices to prove that
 \begin{equation*}
 ||A_{j_0}(\theta -\tilde{q}_k\alpha)||\leq Ce^{\tilde{q}_k(\gamma_\lambda(e)+\varepsilon)}.
 \end{equation*}
 Notice that
   \begin{eqnarray*}
     A_{j_0}(\theta -\tilde{q}_k\alpha) &=& \prod_{j_0-1}^{j=0}A(\theta+j\alpha -\tilde{q}_k\alpha) \\
       &=&  \prod_{j_0-1}^{j=0}\frac{D(\theta+j\alpha-\tilde{q}_k\alpha)}{\cos\pi(\theta+j\alpha-\tilde{q}_k\alpha)}.
   \end{eqnarray*}
 By (\ref{EstimateD}), one has
 \begin{equation}\label{appendixB_1}
     \prod_{j_0-1}^{j=0}D(\theta+j\alpha-\tilde{q}_k\alpha)\leq Ce^{(L+\varepsilon)j_0}.
 \end{equation}
 By (\ref{jneq j_o}),
 \begin{equation*}
   |\cos\pi(\theta+j\alpha )|\geq \frac{1}{C\tilde{q}_k} \text{ for all } j\neq j_0.
 \end{equation*}
 Combing with (\ref{q_k+1}), we have
 \begin{equation}\label{AppendixBG44}
     |\cos\pi(\theta+j\alpha -\tilde{q}_k\alpha)|\geq \frac{1}{C\tilde{q}_k} \text{ for all } j\neq j_0.
 \end{equation}
 If $ |\cos\pi(\theta+j_0\alpha -\tilde{q}_k\alpha)|$ is the smallest one of  $|\cos\pi(\theta+j\alpha -\tilde{q}_k\alpha)|$, $j=0,1,\cdots,\tilde{q}_k-1$,
 applying Lemma \ref{Estimatecos},
one has
\begin{eqnarray}
\nonumber
   \prod_{j=0}^{j_0-1} |\cos\pi(\theta+j\alpha-\tilde{q}_k\alpha)| &\geq& \prod_{j=0,j\neq j_0}^{\tilde{q}_k-1} |\cos\pi(\theta+j\alpha-\tilde{q}_k\alpha)|  \\
   &\geq&  \frac{e^{(     -\ln2-\varepsilon)\tilde{q}_k }}{C\tilde{q}_k}.\label{appendixB_2}
\end{eqnarray}
If for some $j'\neq j_0$,  $ |\cos\pi(\theta+j'\alpha -\tilde{q}_k\alpha)|$ is  the smallest one of  $|\cos\pi(\theta+j\alpha -\tilde{q}_k\alpha)|$, $j=0,1,\cdots,\tilde{q}_k-1$, by Lemma \ref{Estimatecos} and (\ref{AppendixBG44}),  (\ref{appendixB_2}) also holds.
\par
Putting   (\ref{appendixB_1}) and (\ref{appendixB_2}) together,
 \begin{eqnarray*}
   || A_{j_0}(\theta -\tilde{q}_k\alpha)|| &\leq& Ce^{(L+\varepsilon)j_0}\tilde{q}_k e^{(      \ln2+\varepsilon)\tilde{q}_k  }  \\
     &\leq&  Ce^{\tilde{q}_k(\gamma_\lambda(e)+2\varepsilon)}.
 \end{eqnarray*}
 This yields Lemma \ref{appendixLemB_1}.

 \end{proof}
\section{}
\begin{lemma}\label{A.1}

  Let
\begin{equation*}
\{t_k\}_{k=1}^{\infty}=\{q_n, 2q_n, 3q_n,\cdots, \ell_{ n} q_n\} _{n=1}^{\infty} ,
\end{equation*}
   with $  \ell_{ n} = \min\{ \lfloor\frac{ q_n^2}{||q_nk_{\lambda}(e)||_{\mathbb{R}/\mathbb{Z}}}\rfloor, \lfloor \frac{2q_{n+1} }{q_n}  \rfloor\}$, then for some $t>0$, we have
   \begin{equation}\label{appendixG1}
     \sum_{j\neq 0}\sup_{k}|\frac{1-e^{-2\pi ijt_k\alpha}}{1-e^{-2\pi ij\alpha}}  {\zeta} _j |e^{ t |j|}<\infty.
   \end{equation}
\end{lemma}
\begin{proof}
Fix some  $t_k=\ell q_n$ with $\ell\leq \frac{ q_n ^2}{||q_nk_{\lambda}(e)||_{\mathbb{R}/\mathbb{Z}}} $ and $\ell\leq\frac{2q_{n+1} }{q_n} $.
 \par
 If $|j|=q_{n'}$ with $n'\leq n-1$,  by (\ref{Valuezeta_n})
  one has
 \begin{eqnarray*}
    |\frac{1-e^{-2\pi ijt_k\alpha}}{1-e^{-2\pi i j\alpha}}   \zeta  _j | &\leq&  C |j|\ell \frac{ ||q_n\alpha||_{\mathbb{R}/\mathbb{Z}}}{||q_{n'}\alpha||_{\mathbb{R}/\mathbb{Z}}}| \zeta  _j| \\
     &\leq&  C |j|\frac{q_{n+1} }{q_n}    \frac{q_{n'+1}}{q_{n+1}}| \zeta  _j|\\
      &\leq&   C  |j \zeta  _j| \leq  C   e^{-\frac{\gamma_{\lambda}(e)}{2}|j|}.
 \end{eqnarray*}
 If $|j|=q_{n }$, by
   (\ref{Valuezeta_n}) again, one has
 \begin{equation*}
    | \zeta_{j}  | \leq  C ||q_nk_{\lambda}(e)||_{\mathbb{R}/\mathbb{Z}}  e^{-\gamma_{\lambda}(e)|q_n|} .
 \end{equation*}
 Thus
 \begin{eqnarray*}
    |\frac{1-e^{-2\pi ijt_k\alpha}}{1-e^{-2\pi i j\alpha}}   \zeta  _j | &\leq&  C \frac{t_k||j\alpha||_{\mathbb{R}/\mathbb{Z}}}{||j\alpha||_{\mathbb{R}/\mathbb{Z}}}| \zeta  _j| \\
     &\leq&  C\frac{ q_n^3}{||q_n k_{\lambda}(e)||_{\mathbb{R}/\mathbb{Z}}} ||q_n k_{\lambda}(e)||_{\mathbb{R}/\mathbb{Z}}  e^{-\gamma_{\lambda}(e)|q_n|} \\
      &\leq&   C   e^{-\frac{\gamma_{\lambda}(e)}{2}|j|}.
 \end{eqnarray*}
 If $|j|=q_{n'}$ with $n'\geq n+1$,  also by (\ref{Valuezeta_n})
  one has
  \begin{eqnarray*}
    |\frac{1-e^{-2\pi ijt_k\alpha}}{1-e^{-2\pi i j\alpha}}   \zeta  _j | &\leq&  C \frac{t_k||j  \alpha||_{\mathbb{R}/\mathbb{Z}}}{||j \alpha||_{\mathbb{R}/\mathbb{Z}}}| \zeta  _j| \\
     &\leq&  C  q_{n+1}   e^{-\gamma_{\lambda}(e)  q_{n'}}    \\
      &\leq&   C   e^{-\frac{\gamma_{\lambda}(e)}{2}|j|}.
 \end{eqnarray*}
 If $q_{n'-1}<|j|< q_{n'}$ for some $n'\in \mathbb{N}^+$,  by an
 argument similar to the proof of Lemma \ref{Firstpart} ,
 we also have
   \begin{eqnarray*}
    |\frac{1-e^{-2\pi ijt_k\alpha}}{1-e^{-2\pi i j\alpha}}   \zeta  _j | &\leq& | \frac{C\zeta  _j}{1-e^{-2\pi i j\alpha}  }|\\
    &\leq&  C   e^{-\frac{\gamma_{\lambda}(e)}{2}|j|} .
 \end{eqnarray*}
  The  estimate (\ref{appendixG1})   is easy to obtain from the above   cases.

\end{proof}

\begin{lemma}\label{A.2}
Suppose for $n$ large enough, the following holds,
\begin{equation*}
    ||q_n k_{\lambda}(e)  ||_{\mathbb{R}/\mathbb{Z}} \leq C\frac{q_n}{q_{n+1}  }+C\frac{e^{(\delta +\varepsilon)q_n}}{q_{n+1}}.
\end{equation*}
Let $\tilde{\zeta}(x)=\zeta(x)-\zeta_0.$ Then equation
\begin{equation}\label{appendixG2}
   e^{-2\pi i\tilde{\zeta}(x)} \hat{c}_1(x)= \hat{c}_1(x+  \alpha)
\end{equation}
has an analytic solution.
\end{lemma}
\begin{proof}
The proof of this lemma is similar to that of Lemma  \ref{Firstpart}. We only show some key steps.
We will prove that equation (\ref{appendixG2}) has an  analytic solution  of the form $\hat{c_1}(x)=e^{-2\pi i \psi(x) }$, where $\psi(x)$ is real analytic on $\mathbb{T}$.
Substituting  $\hat{c}(x)=e^{-2\pi i \psi(x)}$ into   (\ref{appendixG2}) and comparing the Fourier coefficients,
one has
 \begin{equation*}
    \psi_k=\frac{\zeta _k}{ e^{-  2\pi i k\alpha} -1}, k\neq 0.
 \end{equation*}
 In order to prove this lemma, we only need to prove that $\psi_k$
 is exponentially small.
 \par
  If $|k|=q_n$,  we have
 \begin{eqnarray*}
   |\psi_{k}| &\leq& C\frac{e^{-\gamma_{\lambda}(e)q_n}|| q_n k_{\gamma}(e))||_{\mathbb{R}/\mathbb{Z}}}{|| q_n \alpha||_{\mathbb{R}/\mathbb{Z}}}  \\
     &\leq&  C e^{-(\gamma_{\lambda}(e)-\delta -\varepsilon)q_n}  +C    e^{-\gamma_{\lambda}(e)q_n}   q_n.
 \end{eqnarray*}
  Thus $\psi_k$
 is exponentially decaying.
 \par
 The proofs of the  other cases are the same as in the proof of Lemma \ref{Firstpart}.
\end{proof}
\section*{Acknowledgments}

This research was partially supported by NSF DMS-1401204. This work was done when Wencai  Liu visited University of California,
Irvine in 2014. He would like to thank Xiaoping Yuan (Professor
 at  Fudan University) for supporting him partly. SJ would like to
 thank Mira Shamis for reminding her of the open problems on the Maryland model  and the organizers of Kent
 State Informal Analysis Seminar in March 2014, during which this research was started.
\footnotesize

\bibliography{Maryland3}

\end{document}